\documentclass[letterpaper, 10 pt, conference]{ieeeconf}  

\usepackage{cite}
\usepackage{subcaption}
\usepackage{algorithm}
\usepackage{algpseudocode}
\usepackage{amsmath, amsfonts, amssymb, arydshln}
\usepackage{booktabs}
\usepackage{latexsym}
\usepackage{dsfont}
\usepackage{floatpag}
\usepackage{float}

\usepackage{verbatim}
\newtheorem{theorem}{Theorem}

\usepackage{graphicx,ragged2e}

\IEEEoverridecommandlockouts                              

\title{\LARGE \bf
Strategic Dispatching Equilibrium under Competition in Multi-Regional Ride-Hailing Markets}

\author{
Ran Chen, Nikolas Geroliminis$^*$ 
\thanks{Ran Chen, and Nikolas Geroliminis are with the Urban Transport Systems Laboratory (LUTS), École Polytechnique Fédérale de Lausanne (EPFL), Switzerland.}
\thanks{This work was supported as a part of NCCR Automation, a National Centre of Competence in Research, funded by the Swiss National Science Foundation (grant number 51NF40\_225155).}
\thanks{© 2026 IEEE. Personal use of this material is permitted. Permission from IEEE must be obtained for all other uses, in any current or future media, including reprinting/republishing this material for advertising or promotional purposes, creating new collective works, for resale or redistribution to servers or lists, or reuse of any copyrighted component of this work in other works.}
}

\begin{document}

\maketitle
\thispagestyle{empty}
\pagestyle{empty}

\begin{abstract}
This paper studies competition between ride-hailing companies in a multi-regional market through a static non-cooperative game, in which each company allocates its fleet across regions to maximize profit. Under the assumptions of the proposed framework, we establish an equilibrium existence result on a restricted feasible set by analyzing the asymptotic quasi-concavity of the payoff functions. A numerical study of a two-company, two-region market shows that relative fleet size strongly affects equilibrium dispatching strategies: the smaller company tends to concentrate its fleet in one region and serves both regions only after its fleet share (in percentage) exceeds a threshold. In the numerical setting considered here and under the Logit demand specification, more balanced fleet sizes are associated with higher total industry profit and higher total served demand than strongly asymmetric fleet configurations.

\end{abstract}

\section{Introduction}
Urban transportation systems are rapidly evolving, and ride-hailing (RH) has become a major form of on-demand mobility. By improving the matching between passengers and drivers, RH platforms such as Uber, Lyft, and Didi have substantially changed urban travel markets. At the same time, their growth in competitive and spatially distributed markets has created new challenges for fleet management and market regulation. A key issue in this context is the aggregate management of RH fleets across regions. Dispatching decisions are strongly affected by spatial demand imbalances and local traffic conditions, often causing excessive vehicle concentration in attractive areas, which worsens congestion and reduces system efficiency~\cite{castillo2017surge, simonetto2019real, XuA_2021}. For large-scale systems, such challenges motivate aggregate and control-oriented modeling approaches, including macroscopic formulations~\cite{geroliminis2008existence}. By reducing high-dimensional vehicle-level interactions to tractable aggregate relations, macroscopic models are particularly suitable for large-scale analysis and control. They have been used to study RH operations in multi-regional settings and to design predictive, bi-level, and hierarchical fleet management strategies~\cite{zhu2023data, beojone2024hierarchical}. Other RH models have examined matching frictions, surge pricing, driver labor supply, and platform competition from economic and market-equilibrium perspectives~\cite{chen2016dynamic}.

Competition between mobility operators is commonly analyzed using game-theoretic models~\cite{fisk1984game, ban2019general}. In the RH literature, prior studies have examined labor supply responses to flexible platform work~\cite{hall2018analysis} and inter-platform competition in regulated ride-hail markets with pooling and single- or multi-homing supply~\cite{zhang2021inter}. Related work has also extended competition analysis to operational settings such as hierarchical pricing for charging ride-hailing electric fleets~\cite{maljkovic2023hierarchical}. More broadly, recent studies have considered richer competitive environments involving multiple mobility service providers and multimodal equilibrium interactions~\cite{najmi2023multimodal, bandiera2024mobility}.

Against this background, we adopt a stylized macroscopic setting to isolate the effect of fleet size heterogeneity on regional dispatching behavior. Specifically, we study how the relative fleet sizes of competing RH companies affect their aggregate dispatching strategies in a multi-regional market. We formulate the interaction as a non-cooperative game in which each company allocates its fleet across regions to maximize its revenue. Within this simplified framework, our goal is to characterize how market structure shapes dispatching incentives and equilibrium outcomes. On the theoretical side, we analyze equilibrium existence on a restricted feasible set by studying the asymptotic quasi-concavity property of the payoff functions. On the numerical side, we use a two-company, two-region example to illustrate how relative fleet size affects equilibrium strategies and aggregate outcomes. In this case study, the numerical results suggest that fleet parity is associated with higher total profit and served demand under the Logit-based demand specification.

The remainder of this paper is organized as follows. Section~\ref{sec:model} introduces the game formulation and its main components. Section~\ref{sec:existence} presents the equilibrium analysis. Section~\ref{sec:result} reports the numerical study on a simplified system. Finally, Section~\ref{sec:concl} concludes the paper and discusses directions for future research.
\section{Static Dispatching Game}
\label{sec:model}

Consider a ride-hailing market with $n$ competing companies and $m$ regions, each characterized by several local properties. Each company $i\in \mathcal{N}=\{1,2,\cdots\}$ operates a total fleet ${N}_i$, distributed across regions $j\in \mathcal{M}=\{A,B,\cdots\}$ as ${N}_i^j$, with
\begin{equation}
\label{eq:sumNij}
    \sum_{j\in\mathcal{M}} N_i^j = {N}_i. 
\end{equation}
For simplicity, we model ${N_i^j}$ as a continuous aggregate fleet allocated by company $i$ to region $j$. Similarly, ${V_i^j}$, ${O_i^j}$, and ${D_i^j}$ denote the empty fleet, occupied fleet in [veh], and demand inflow rate in [veh/min] in region $j$. At steady state, these quantities satisfy the regional conservation equation inspired by \cite{xu2020supply}:
\begin{equation}
\label{eq:N=V+(w+T)D}
N_i^j = V_i^j + (w_i^j + T_i^j) D_i^j = V_i^j + O_i^j,
\end{equation}
where ${w_i^j}$ is the expected passenger waiting time, $T_i^j$ is the in-vehicle travel time, and $O_i^j$ denotes the occupied fleet. To simplify the system, we adopt the following assumptions:
\begin{enumerate}
\item \textbf{Stationary regional conditions:} For each region, the travel time $T_i^j$ and total company fleet ${N}_i=\sum_jN_i^j$ remain constant over the analyzed period.
\item \textbf{Monotonic fleet decomposition:} Both $V_i^j$ and $O_i^j$ increase with the dispatched fleet $N_i^j$.
\item \textbf{Cost-sensitive demand:} The realized demand $D_i^j$ depends on the generalized costs perceived by customers.
\item \textbf{Accessible information:} Companies observe the generalized costs of their competitors.
\end{enumerate}

Assumption 1 reflects the steady-state and regional nature of the model. Assumption 2 characterizes an efficient operating regime in which additional fleet is productively used \cite{kaddoura2021impact}, implying
\begin{equation}
     \frac{\partial V_i^j}{\partial N_i^j}+\frac{\partial O_i^j}{\partial N_i^j} =1,
     \qquad
     \frac{\partial V_i^j}{\partial N_i^j}>0,
     \qquad
     \frac{\partial O_i^j}{\partial N_i^j}>0.
\end{equation}

To model customer choice under Assumption 3, let $-i$ denote the set of competitors other than company $i$. We adopt a linear generalized cost, exponential utility, and a Logit demand model:
\begin{equation}
\label{eq:C}
C_i^j(w_i^j,T_i^j,f_i^j)= \lambda_w^j w_i^j+\lambda_T^jT_i^j+\lambda_f^j f_i^j,
\end{equation}
\begin{equation}
    u_i^j(C_i^j)=\exp(-\theta^jC_i^j),
\end{equation}
\begin{equation}
    D_i^j(u_i^j, u_{-i}^j)= D_\mathrm{max}^j\frac{u_i^j}{u_\mathrm{PT}^j+u_i^j+u_{-i}^j}
    =D_\mathrm{max}^j\frac{u_i^j}{u_\mathrm{tot}^j}.
\end{equation}
For simplicity, the fare $f_i^j$ is treated as fixed, reflecting that pricing decisions are typically updated on a slower timescale than dispatching decisions. Here, $D_\mathrm{max}^j$ is the maximum realizable demand shared by ride-hailing and public transport, $\lambda$ and $\theta$ are user sensitivity parameters, and $u_\mathrm{PT}^j$ is the utility of public transport. Assumption 4 is natural since waiting times and fares are typically observable to customers. Finally, we express the waiting time as a convex decreasing function of the empty fleet:
\begin{equation}
\label{eq:w}
    w_i^j(V_i^j)=\beta^j (V_i^j)^{-c^j},
\end{equation}
where typically $c^j=0.5$ \cite{arnott1996taxi}. Hence, $D_i^j$ is ultimately a function of $V_i^j$ and $V_{-i}^j$. For conciseness, these arguments are omitted hereafter. Then, under Assumption 2, the occupied fleet $O_i^j$ satisfies
\begin{equation}
\label{eq:dOijdNij>0}
    \frac{\partial O_i^j}{\partial N_i^j}
    =D_i^j\frac{\partial w_i^j}{\partial V_i^j}\frac{\partial V_i^j}{\partial N_i^j}
    +(w_i^j+T_i^j)\frac{\partial D_i^j}{\partial N_i^j} >0.
\end{equation}
This implies that demand is positively responsive:
\begin{equation}
\label{eq:dDijdNij>0}
    \frac{\partial w_i^j}{\partial V_i^j}<0<\frac{\partial V_i^j}{\partial N_i^j}
    \quad\Rightarrow\quad
    \frac{\partial D_i^j}{\partial N_i^j}
    >-\frac{D_i^j\frac{\partial w_i^j}{\partial V_i^j}\frac{\partial V_i^j}{\partial N_i^j}}{w_i^j+T_i^j}
    >0.
\end{equation}
Conversely, for a competitor $r\neq i$,
\begin{equation}
\begin{aligned}
    \frac{\partial D_r^j}{\partial N_i^j}
    =-D_\mathrm{max}^j\frac{u_r^j}{(u_\mathrm{tot}^j)^2}
    \frac{\partial u_i^j}{\partial V_i^j}\frac{\partial V_i^j}{\partial N_i^j}<0.
\end{aligned}
\end{equation}
The corresponding response of competitor $r$'s occupied fleet is
\begin{equation}
\label{eq:dOrjoverdNij}
\begin{aligned}
    \frac{\partial O_r^j}{\partial N_i^j}
    &= D_r^j\frac{\partial w_r^j}{\partial V_r^j}\frac{\partial V_r^j}{\partial N_i^j}
      +(w_r^j+T_r^j)\frac{\partial D_r^j}{\partial N_i^j} \\
    &= D_r^j\frac{\partial w_r^j}{\partial V_r^j}\left(-\frac{\partial O_r^j}{\partial N_i^j}\right)
      +(w_r^j+T_r^j)\frac{\partial D_r^j}{\partial N_i^j} \\
    &= \frac{(w_r^j+T_r^j)\frac{\partial D_r^j}{\partial N_i^j}}
    {1+D_r^j\frac{\partial w_r^j}{\partial V_r^j}}.
\end{aligned}
\end{equation}
We seek $\partial O_r^j/\partial N_i^j<0$, meaning that an increase in the allocated fleet of company $i$ reduces the occupied fleet of competitor $r$: as company $i$ improves its service level, competitor $r$ loses demand, and its occupied fleet should decrease accordingly as a competitive response. In Eq.~\eqref{eq:dOrjoverdNij}, the numerator is strictly negative as $\partial D_r^j/\partial N_i^j<0$; it remains to ensure that the denominator is positive. To this end, define the threshold ${V}_{r,\mathrm{inv}}^j$ by
\begin{equation}
\label{eq:Vrjcr}
    \left.\frac{\partial w_r^j}{\partial V_r^j}\right|_{{V}_{r,\mathrm{inv}}^j}
    =-\frac{1}{D_\mathrm{max}^j},
\end{equation}
which exists as Eq.~\eqref{eq:w} is adopted. Then, $\forall \,V_r^j>{V}_{r,\mathrm{inv}}^j$,
\begin{equation}
\label{eq:posi_deno}
    1+D_r^j\left.\frac{\partial w_r^j}{\partial V_r^j}\right|_{{V}_{r,\mathrm{inv}}^j}
    >1-\frac{D_r^j}{D_\mathrm{max}^j}>0.
\end{equation}
Using Eq.~\eqref{eq:N=V+(w+T)D}, this threshold induces a lower bound $N_{r,\mathrm{inv}}^j$, thus in $N_r^j\in[N_{r,\mathrm{inv}}^j, N_r]$, 
\begin{equation}
\label{eq:dVrj+dOrj=0}
    \frac{\partial V_r^j}{\partial N_i^j}+\frac{\partial O_r^j}{\partial N_i^j}
    = \frac{\partial N_r^j}{\partial N_i^j} = 0,
    \,
    \frac{\partial O_r^j}{\partial N_i^j} = -\frac{\partial V_r^j}{\partial N_i^j} <0.
\end{equation}
Now introduce the collective fleet vectors
\begin{equation}
    \mathbf{N}=[\mathbf{N}_1^\top,\cdots, \mathbf{N}_n^\top]^\top,
    \,
    \mathbf{N}_i = [N_i^A, N_i^B, \cdots]^\top,
\end{equation}
and similarly for $\mathbf{V}$. The Jacobians
\begin{equation}
    \label{eq:jacobian_v_n}
    \mathbf{J}_\mathbf{V}(\mathbf{N})  = \frac{\partial\mathbf{V}}{\partial\mathbf{N}},\, \mathbf{J}_\mathbf{N}(\mathbf{V})  = \frac{\partial\mathbf{N}}{\partial\mathbf{V}},
\end{equation}
will be used later to ensure that $\mathbf{V}(\mathbf{N})$ is a well-defined  function. Let $\mathbf{P}$ be the orthogonal permutation matrix that groups variables by region:
\begin{equation}
\tilde{\mathbf{N}}=\mathbf{P}\mathbf{N}=\mathrm{col}\{\mathbf{N}^j\}_{j\in\mathcal{M}},
\,
\mathbf{N}^j=\mathrm{col}\{{N}_i^j\}_{i\in\mathcal{N}},
\end{equation}
and similarly $\tilde{\mathbf{V}} = \mathbf{P}\mathbf{V}$. Then
\begin{equation}
    \mathbf{J}_{\tilde{\mathbf{N}}}(\tilde{\mathbf{V}})
    = \mathbf{P}\mathbf{J}_\mathbf{N}(\mathbf{V})\mathbf{P}^\top.
\end{equation}
Since the dispatched fleet in each region depends only on the empty fleets in that same region, the cross-region derivatives vanish, yielding the block-diagonal form
\begin{equation}
    \mathbf{J}_{\tilde{\mathbf{N}}}(\tilde{\mathbf{V}})
    =
    \mathrm{diag}\left\{\mathbf{J}^j_{\tilde{\mathbf{N}}}(\tilde{\mathbf{V}}^j)\right\}_{j},
\end{equation}
where each regional block is
\begin{equation}
\label{eq:jacobian_general}
\mathbf{J}^j_{\tilde{\mathbf{N}}}(\tilde{\mathbf{V}}^j)
=
\frac{\partial \mathbf{N}^j}{\partial \mathbf{V}^j}
=
\begin{bmatrix}
\frac{\partial N_1^j}{\partial V_1^j} &\cdots & \frac{\partial N_1^j}{\partial V_n^j} \\
\vdots&\ddots & \vdots \\
\frac{\partial N_n^j}{\partial V_1^j} & \dots & \frac{\partial N_n^j}{\partial V_n^j}
\end{bmatrix}.
\end{equation}
Hence, the existence of $\mathbf{J}_\mathbf{V}(\mathbf{N})$ reduces to the invertibility of these regional blocks. The Jacobian matrix $\mathbf{J}_\mathbf{V}(\mathbf{N})$ exists for $\mathbf{N}\ge\mathbf{N}_{\mathrm{inv}}$. The proof is given in Appendix~\ref{appen:invproof}.

\section{Existence of Nash Equilibrium}
\label{sec:existence}

We now formulate the Nash game and analyze equilibrium existence using standard conditions on the feasible set and payoff functions \cite{rosen1965existence}. Given $\mathbf{N}_{-i}$ the fleet allocations of competitors, company $i$ solves
\begin{equation}
\label{eq:game_prob}
\begin{aligned}
    \mathcal{P}_i:\;\max_{\mathbf{N}_i\geq \mathbf{0}} & \quad p_i(\mathbf{N}_i, \mathbf{N}_{-i}) \\
   \mathrm{s.t.} &\quad 
   \sum_j N_{i}^j \leq {N}_i.
\end{aligned}
\end{equation}
The payoff\footnote{The operating cost is omitted to isolate the fleet-allocation mechanism under fixed fares. Over the dispatching timescale considered here, such costs can be treated as region-invariant or absorbed into a baseline term, and therefore do not alter the asymptotic concavity argument.} per unit time in [\$/min] is defined as
\begin{equation}
    p_i = \sum_{j\in \mathcal{M}}f_i^jT_i^jD_i^j(\mathbf{V}^j(\mathbf{N}^j)).
\end{equation}
The Nash equilibrium is defined as $\mathbf{N}^\star$ such that
\begin{equation}
\label{eq:gne}
    p_i(\mathbf{N}_i^\star, \mathbf{N}_{-i}^\star)\geq \max_{\mathbf{N}_i}\left\{ p_i(\mathbf{N}_i, \mathbf{N}_{-i}^\star)\;|\;\mathbf{N}_i\in\mathcal{X}_i\right\},
\end{equation}
for each $i$ where
\begin{equation}
\label{eq:def_chi}
\mathcal{X}_{i} = \left\{\mathbf{N}_i\;|\;{N}_i^j\geq 0,\sum_j N_{i}^j \leq {N}_i\right\}.
\end{equation}
Since $\mathbf{V}^j(\mathbf{N}^j)$ is defined implicitly, the payoff $p_i$ generally has no closed-form expression and is not globally concave. We therefore focus on the \textbf{asymptotic quasi-concavity} of $p_i$ by showing that there exists a threshold $N_{i,\mathrm{ac}}^j$ such that
\begin{equation}
\label{eq:d2pidNi2}
\frac{\partial^2 p_i}{\partial (N_i^j)^2}
=f_i^j T_i^j\frac{\partial^2 D_i^j}{\partial (N_i^j)^2 }<0
\quad \forall N_i^j>N_{i,\mathrm{ac}}^j,
\end{equation}
for fixed $N_{-i}^j$. Hence, the curvature of $p_i$ is directly determined by that of $D_i^j$. To establish this property, we first examine the first derivative of $D_i^j$:
\begin{equation}
    \begin{aligned}
    \label{eq:full_dDidNi}
        \frac{\partial D_i^j}{\partial N_i^j} 
        & = \frac{\partial D_i^j}{\partial V_i^j}\frac{\partial V_i^j}{\partial N_i^j}
        + \sum_{r\neq i}\frac{\partial D_i^j}{\partial V_r^j}\frac{\partial V_r^j}{\partial N_i^j}.
    \end{aligned}
\end{equation}
To evaluate $V_i^j$ and $V_r^j$ w.r.t. $N_i^j$, consider
\begin{equation}
     \frac{\partial \mathbf{N}^j}{\partial \mathbf{V}^j}
     =\mathbf{I}+\frac{\partial \mathbf{O}^j}{\partial \mathbf{V}^j}
     =\mathbf{I}+\mathbf{M}^j(\mathbf{V}^j),
\end{equation}
where $\mathbf{I}$ is the identity matrix. Since
\begin{equation}
    \sup_{\mathbf{V}^j} O_i^j=\frac{T_i^jD_\mathrm{max}^j}{u_\mathrm{PT}^j+n},
\end{equation}
$\mathbf{M}^j(\mathbf{V}^j)$ decays to zero as $\mathbf{V}^j$ increases. Hence, there exists a threshold $\mathbf{V}_{\mathrm{norm}}^j$ allowing the Neumann series expansion:
\begin{equation}
\label{eq:neumann}
    \frac{\partial \mathbf{V}^j}{\partial \mathbf{N}^j}    =\mathbf{I}+\sum_{k=1}^\infty(-1)^k(\mathbf{M}^j)^k\quad \forall\mathbf{V}^j>\mathbf{V}_\mathrm{norm}^j.
\end{equation}
Then for fixed ${V}_{-i}^j$, we have the following approximation:
\begin{equation}
\label{eq:dVdNelement}
    \begin{aligned}
            \frac{\partial V_i^j}{\partial N_i^j}&=1-O\left(\frac{\partial O_i^j}{\partial V_i^j}\right),\\
            \frac{\partial V_r^j}{\partial N_i^j}&=-\frac{\partial O_r^j}{\partial V_i^j}+O\left( \left( \frac{\partial O_r^j}{\partial V_i^j} \right)^2  \right),\quad \forall r\neq i,
    \end{aligned}
\end{equation}
where
\begin{equation}
\label{eq:dOidVi}
    \begin{aligned}
        \frac{\partial O_i^j}{\partial V_i^j}
        & = D_\mathrm{max}^j\frac{u_i^j}{u_\mathrm{tot}^j} \frac{\partial w_i^j}{\partial V_i^j}\\
        &\quad + (w_i^j+T_i^j)D_\mathrm{max}^j\frac{u_\mathrm{tot}^j-u_i^j}{(u_\mathrm{tot}^j)^2}\frac{\partial u_i^j}{\partial V_i^j},
    \end{aligned}
\end{equation}
\begin{equation}
\label{eq:dOrdVi}
    \frac{\partial O_r^j}{\partial V_i^j}
    =-(w_r^j+T_r^j)D_\mathrm{max}^j\frac{u_r^j}{(u_\mathrm{tot}^j)^2}\frac{\partial u_i^j}{\partial V_i^j}.
\end{equation}
Applying Eqs.~\eqref{eq:C}--\eqref{eq:w} gives
\begin{equation}
\label{eq:dwijdVij}
    \frac{\partial w_i^j}{\partial V_i^j}= -c^j\beta^j(V_i^j)^{-(c^j+1)}
    \sim O\left( (V_i^j)^{-(c^j+1)}\right),
\end{equation}
\begin{equation}
\label{eq:duijdVij}
    \frac{\partial u_i^j}{\partial V_i^j}= \lambda_w^j\beta^ju_i^j\cdot(V_i^j)^{-(c^j+1)}
    \sim O\left( (V_i^j)^{-(c^j+1)}\right),
\end{equation}
and therefore
\begin{equation}
    \frac{\partial O_i^j}{\partial V_i^j},\;\frac{\partial O_r^j}{\partial V_i^j}
    \sim O\left( (V_i^j)^{-(c^j+1)}\right).
\end{equation}
Since $c^j>0$, the higher-order terms in Eq.~\eqref{eq:dVdNelement} vanish as $V_i^j$ increases. Similarly,
\begin{equation}
\label{eq:dDidVi}
    \frac{\partial D_i^j}{\partial V_i^j}
    =D_\mathrm{max}^j\frac{u_\mathrm{tot}^j-u_i^j}{(u_\mathrm{tot}^j)^2}\frac{\partial u_i^j}{\partial V_i^j}
    \sim O\left(  (V_i^j)^{-(c+1)}\right),
\end{equation}
\begin{equation}
\label{eq:dDidVr}
    \frac{\partial D_i^j}{\partial V_r^j}
    =-D_\mathrm{max}^j\frac{u_i^j}{(u_\mathrm{tot}^j)^2}\frac{\partial u_r^j}{\partial V_r^j}
    \sim O\left(  (V_r^j)^{-(c+1)}\right).
\end{equation}
Substituting Eqs.~\eqref{eq:dOrdVi}, \eqref{eq:dDidVi}, and \eqref{eq:dDidVr} into Eq.~\eqref{eq:full_dDidNi} yields
\begin{equation}
    \begin{aligned}
    \label{eq:11}
        \frac{\partial D_i^j}{\partial N_i^j} 
        & =\frac{\partial D_i^j}{\partial V_i^j}
        - \sum_{r\neq i}\frac{\partial D_i^j}{\partial V_r^j}\frac{\partial O_r^j}{\partial V_i^j}
        + O\left( (V_i^j)^{-2(c+1)}\right)\\ 
        & \approx \frac{D_\mathrm{max}^j}{u_\mathrm{tot}^j}\left( 1-\frac{{u}_i^j}{u_\mathrm{tot}^j}-\frac{{u}_i^j}{(u_\mathrm{tot}^j)^3} K(V_{-i}^j)\right)\frac{\partial u_i^j}{\partial V_i^j} \\
        & = E(u_i^j)\frac{\partial u_i^j}{\partial V_i^j},
    \end{aligned}
\end{equation}
where
\begin{equation}
    K(V_{-i}^j)= D_\mathrm{max}^j\sum_{r\neq i}(w_r^j+T_r^j)u_r^j\frac{\partial u_r^j}{\partial V_r^j}.
\end{equation}
Since $K(V_{-i}^j)$ is independent of $V_i^j$, $E(u_i^j)$ converges to a constant. Computing the second derivative of $u_i^j$ gives
\begin{equation}
\label{eq:d2uidVi2}
    \begin{aligned}
         \frac{\partial ^2u_i^j}{\partial (V_i^j)^2} 
         &= c^j(c^j+1) \lambda_w^j\beta^j u_i^j\cdot (V_i^j)^{-(c^j + 2)}\\
         &\quad \cdot \left( \frac{c^j\lambda_w^j\beta^j}{c^j+1}(V_i^j)^{-c^j}   -1\right).
    \end{aligned}
\end{equation}
Hence,
\begin{equation}
\label{eq:concaveui}
       \frac{\partial ^2u_i}{\partial (V_i^j)^2}<0
       \qquad\forall 
V_i^j>\left(  \frac{c^j\lambda_w^j\beta^j}{c^j+1}\right)^{1/{c^j}},
\end{equation}
implying the existence of a lower bound of $V_i^j$ such that
\begin{equation}
     \frac{\partial^2 u_i^j}{(\partial V_i^j)^2}<0\Rightarrow\frac{\partial^2 D_i^j}{(\partial N_i^j)^2}<0\Rightarrow\frac{\partial^2 p_i^j}{(\partial N_i^j)^2}<0.
\end{equation}
Thus, one can define the Restricted Nash Equilibrium (RNE) $\mathbf{N}_{\mathrm{res}}^\star$ such that, for all $i$,
\begin{equation}
\label{eq:restrict_neq}
\mathbf{N}_{i,\mathrm{res}}^\star=\arg \max_{\mathbf{N}_i}\left\{ p_i(\mathbf{N}_i, \mathbf{N}_{-i,\mathrm{res}}^\star)\;|\;\mathbf{N}_i\in \mathcal{X}_{i,\mathrm{ac}}\right\},
\end{equation}
and the existence is guaranteed by Theorem 1.
\begin{theorem}
\label{thm:existence}
Consider a competitive ride-hailing market with a company set $\mathcal{N}$. There exists a finite threshold fleet $\mathbf{N}_{\text{ac}}$ defining the collective feasible set $\mathcal{X}_{\mathrm{ac}} = \prod_{i\in\mathcal{N}} \mathcal{X}_{i,\mathrm{ac}}$, where
\begin{equation}
\label{eq:def_chi_ac}
\mathcal{X}_{i,\mathrm{ac}} = \left\{\mathbf{N}_i\;|\;{N}_i^j\geq N_{i,\mathrm{ac}}^j,\sum_j N_{i}^j \leq {N}_i\right\},
\end{equation}
such that $p_i(\mathbf{N}_i, \mathbf{N}_{-i})$ is quasi-concave for all $\mathbf{N}_i\in\mathcal{X}_{i,\mathrm{ac}}$ given fixed $\mathbf{N}_{-i}\in \mathcal{X}_{-i,\mathrm{ac}}$. Consequently, a Restricted Nash Equilibrium $\mathbf{N}^\star = (\mathbf{N}_1^\star, \dots, \mathbf{N}_n^\star)$ exists in the game $\mathcal{G} = (\mathcal{N}, p_i, \mathcal{X}_{i,\mathrm{ac}})$.
\end{theorem}

\begin{proof}
The proof follows the sufficient conditions in \cite{rosen1965existence}. The set $\mathcal{X}_{i,\mathrm{ac}}$ is non-empty provided the threshold is admissible, i.e.,
\begin{equation*}
    \sum_{j} N_{i,\mathrm{ac}}^j \leq N_i.
\end{equation*}
Under this condition, $\mathcal{X}_{i,\mathrm{ac}}$ is defined by finitely many linear inequalities and is therefore non-empty, convex, and compact. The payoff $p_i$ is continuous because $D_i^j$ and $\mathbf{V}$ are continuous in their arguments, while $\mathbf{J}_\mathbf{V}(\mathbf{N})$ exists for sufficiently large $\mathbf{N}$. Finally, Eqs.~\eqref{eq:d2pidNi2}--\eqref{eq:concaveui} establish the asymptotic quasi-concavity of $p_i$ for fixed $\mathbf{N}_{-i}$. Hence, a Restricted Nash equilibrium exists for all $\mathbf{N}\in \mathcal{X}_{\mathrm{ac}}$.
\end{proof}


\section{Numerical Results}
\label{sec:result}

To illustrate the properties of the proposed model, we consider a simplified market with two ride-hailing companies, $i\in\mathcal{N}=\{1,2\}$, and two regions, $j\in \mathcal{M} = \{A, B\}$. This section examines how fleet distribution and market competition affect dispatching strategies and market outcomes. We impose an upper bound on the total number of ride-hailing vehicles, namely the \textbf{total supply level} $N_{\mathrm{tot}}={N}_1+{N}_2$, to represent a regulatory cap intended to mitigate congestion and market saturation \cite{yu2020balancing, chen2024operations}. The simulation parameters are summarized in Table~\ref{tab:simulation_parameters_revised}, with Region A representing a congested city center.


To compare results across different values of ${N}_{\mathrm{tot}}$, the following equilibrium quantities are normalized and expressed in percentage, denoted by the subscript $\%$:
\begin{itemize}
    \item $\mathbf{N}^\star_\% =  [N_{1,\%}^{A,\star}, N_{1,\%}^{B,\star}, N_{2,\%}^{A,\star}, N_{2,\%}^{B,\star}]^\top$: normalized by ${N}_i$.
    \item $\mathbf{V}^\star_\% =  [V_{1,\%}^{A,\star}, V_{1,\%}^{B,\star}, V_{2,\%}^{A,\star}, V_{2,\%}^{B,\star}]^\top$: normalized by ${N}_i$.
    \item $\mathbf{D}^\star_\% =  [D_{1,\%}^{A,\star}, D_{1,\%}^{B,\star}, D_{2,\%}^{A,\star}, D_{2,\%}^{B,\star}]^\top$: normalized by $D_{\mathrm{max}}^j$.
\end{itemize}

\begin{table}
\centering
\caption{Simulation parameters for Regions A and B. Values are inspired by real-world operations.}
\label{tab:simulation_parameters_revised}
\begin{tabular}{@{}llc@{}}
\toprule
\textbf{Notation} & \textbf{Description} & \textbf{Value (A, B) and Unit}  \\
\midrule
$T$ & Average travel time & $9$, $7.5$  [min] \\
$D_{\mathrm{max}}$ & Maximum demand flux & $10$, $5$  [veh/min] \\
$\beta_\mathrm{pickup}$ & Pickup time factor & $30$, $25$  [min] \\
$w_\mathrm{PT}$ & PT waiting time & $10$, $15$  [min] \\
$T_\mathrm{PT}$ & PT traveling time & $10$, $15$  [min] \\
$\theta$ & Cost sensitivity & $1$, $1$  [-] \\
$\lambda_w$ & Value of waiting time & $0.2$, $0.2$  [\$/min$^2$] \\
$\lambda_T$ & Value of travel time & $0.1$, $0.1$  [\$/min$^2$] \\
$\lambda_f$ & Value of charged fare & $1$, $1$  [-] \\
$\mathbf{f}_\mathrm{def}$ & Definite charged fare & $1.8$, $2.2$  [\$/min] \\
\bottomrule
\end{tabular}
\end{table}

\subsection{Variation of Equilibrium Strategy}
\label{subsubsec:NwithNtot}

We investigate the equilibrium as a function of \textbf{Company~1's fleet fraction of the total supply}, defined as
\begin{equation}
    \alpha_1 \cong {N}_1/{N}_{\mathrm{tot}} \times 100\%,
\end{equation}
over the full range from $0\%$ to $100\%$. The analysis is conducted for three supply levels: ${N}_{\mathrm{tot}} = 500, 300$, and $200$. Figure~\ref{fig:Nijcomp} shows the corresponding equilibrium dispatching strategies. The curves are distinguished by company (blue for Company~1, red for Company~2) and by region (solid for A, dashed for B).

We first consider the case $N_{\mathrm{tot}}=500$ (left of Figure~\ref{fig:Nijcomp}). When Company~1 has a relatively small fleet ($\alpha_1 \leq 7.6\%$), its equilibrium strategy is to allocate all vehicles to Region B, i.e., $N_{1,\%}^{B,\star}=100\%$, thereby avoiding direct competition with the larger Company~2 in Region A. Over this range, Company~2 is only weakly affected by Company~1. Once $\alpha_1 > 7.6\%$, Company~1 starts to serve both regions. Beyond $10\%$, it allocates a larger share to Region A than to Region B, and this tendency continues as its fleet fraction increases. The decrease in $N_{1,\%}^{B,\star}$ should be interpreted as a proportional shift: newly added vehicles are increasingly directed to Region A. As $\alpha_1 \to 100\%$, the roles reverse and Company~2 becomes the smaller competitor, exhibiting the same concentration pattern.

The middle panel with $N_{\mathrm{tot}}=300$ shows the same qualitative behavior, but the transition to service in both regions occurs later, at $\alpha_1=13.7\%$. When ${N}_{\mathrm{tot}}=200$, this threshold rises further to $22.9\%$ (right panel). Hence, as the total market supply decreases, the smaller company remains concentrated in one region over a wider range of fleet fractions. The resulting transition is sharper at low supply levels and appears as visible kinks in the curves, especially for ${N}_{\mathrm{tot}}=200$.

\begin{figure*}
    \centering
    \includegraphics[width=0.96\linewidth]{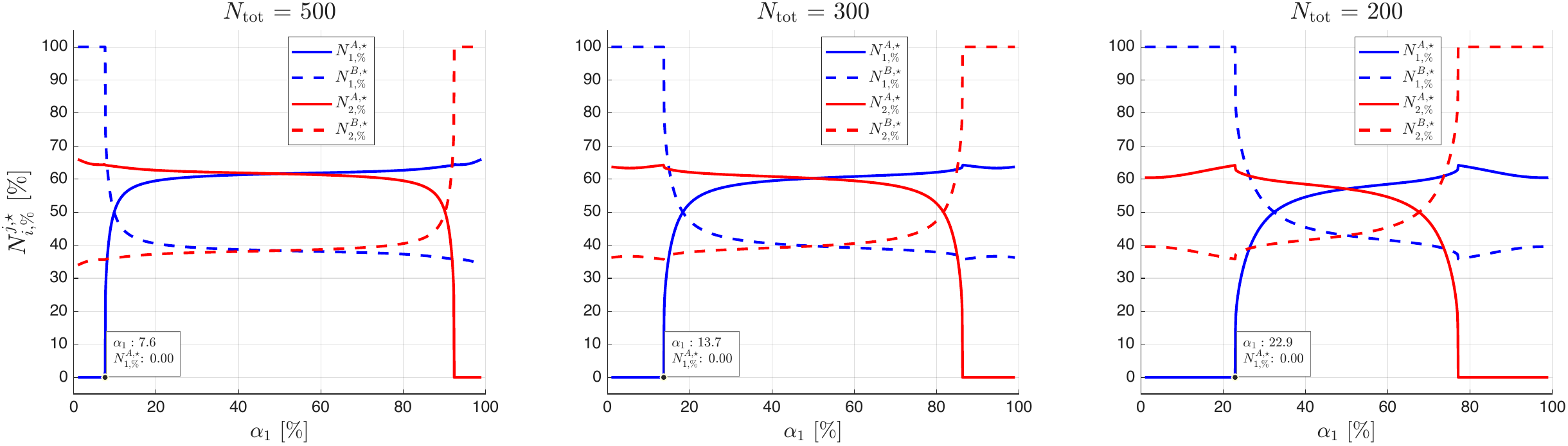}
    \caption{Equilibrium dispatch strategy ($\mathbf{N}_\%^\star$), normalized by ${N}_i$, for different values of ${N}_{\mathrm{tot}}$. Kinks indicate sharp transitions between both-region operation and single-region concentration.}
    \label{fig:Nijcomp}
\end{figure*}

\subsection{Equilibrium Financial and Social Outcomes}

We next analyze financial and service outcomes as functions of $\alpha_1$. Since each company maximizes its own payoff $p_i$, we first examine the equilibrium profits $p_i^\star$. We then evaluate service outcomes through the total realized regional demand, $\mathbf{D}_{\%}^{j,\star} = D_{1,\%}^{j,\star} + D_{2,\%}^{j,\star}$. Results are shown in Figure~\ref{fig:pDcomp}.

The left panel shows the equilibrium profit of each company. As expected, $p_1^\star$ generally increases with $\alpha_1$, whereas $p_2^\star$ decreases, reflecting the competitive nature of the market. The middle panel reports the total industry profit $p^\star =\sum_i p_i^\star$. In all three supply scenarios, the curve is concave and reaches its maximum near $\alpha_1=50\%$, while lower values are observed as the market approaches monopoly.

The right panel shows the normalized total realized demand in each region. Across all ${N}_{\mathrm{tot}}$ scenarios, the ride-hailing market captures more demand in Region B than in Region A. This suggests that Region A is harder to serve efficiently, due to stronger congestion and stronger competition from public transport. The small increase in $D_{\%}^{B,\star}$ for very small or very large $\alpha_1$ is associated with the regional concentration strategy adopted by the smaller competitor. Despite these local variations, the total realized demand (summing up both regions)
is highest near $\alpha_1=50\%$ for all three values of $N_\mathrm{tot}$. Therefore, in this numerical setting and under the adopted Logit demand specification, more balanced fleet sizes are associated with both higher total profit and higher total served demand.

\begin{figure*}
    \centering
    \includegraphics[width=0.96\linewidth]{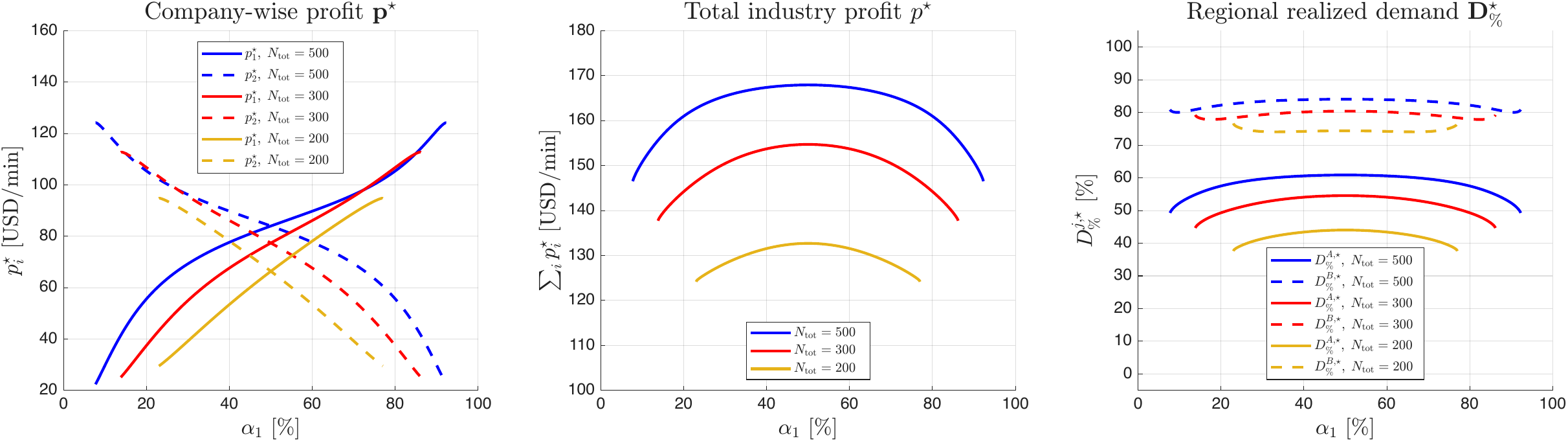}
    \caption{Equilibrium market outcomes as functions of the fleet fraction $\alpha_1$. Single-region concentration solutions with any $N_i^j=0$ are omitted for clarity.}
    \label{fig:pDcomp}
\end{figure*}
\section{Conclusion}
\label{sec:concl}
This paper develops a stylized model of a non-cooperative ride-hailing (RH) market in which competing companies dispatch their fleets across multiple regions to maximize individual profit. The problem is formulated as a static multi-player game. The model is built on an aggregate flow conservation equation, which leads to a nontrivial and generally non-concave payoff. Nevertheless, we establish an equilibrium existence result on a restricted feasible set by proving the asymptotic quasi-concavity of the payoff functions beyond an operational fleet threshold.

A numerical study illustrates the variation of dispatching strategies in a two-company, two-region market with distinct regional characteristics. We impose an upper bound on the total number of vehicles, interpreted as a regulated market supply level, and examine how the equilibrium depends on the fleet shares of the two companies. The results show that the larger company tends to serve both regions, whereas the smaller company tends to concentrate its fleet in one region and begins to serve both regions only after its fleet share percentage exceeds a threshold. This threshold increases as the total supply level decreases. In the numerical setting considered here and under the adopted Logit demand model, more balanced fleet sizes are also associated with higher total industry profit and higher total served demand.

Several directions remain for future research. First, the equilibrium analysis can be strengthened by relaxing some of the current assumptions and identifying conditions for uniqueness. Second, the model can be expanded to include additional decision variables, such as fares, in order to study the joint effects of pricing and dispatching on market and welfare outcomes. Finally, future work can move beyond the static setting by incorporating time-varying or uncertain demand and extending the framework toward dynamic control and repositioning problems.
\section{Appendix}
\subsection{Proof of Invertible Jacobian $\mathbf{J}_\mathbf{N}(\mathbf{V})$}
\label{appen:invproof}

    We have shown the equivalence between the existence of $\mathbf{J}_\mathbf{V}(\mathbf{N})$ and the invertibility of $\mathbf{J}^j_{\tilde{\mathbf{N}}}(\tilde{\mathbf{V}}^j)$. We apply Theorem (4,3) of \cite{fiedler1962matrices}, which states that a matrix ${A}$ belonging to the class $\mathbf{Z}$ (real square matrices with non-positive off-diagonal entries) is invertible if and only if there exists a vector $\mathbf{x} >\mathbf{0}$ such that ${A}\mathbf{x} > \mathbf{0}$. Obviously, $\mathbf{J}^j_{\tilde{\mathbf{N}}}(\tilde{\mathbf{V}}^j)\in\mathbf{Z}$ since its off-diagonals satisfy
    \begin{equation}
    \label{eq:dNijoverdVrj}
     \frac{\partial N_r^j}{\partial V_i^j} 
     = (w_r^j+T_r^j)\frac{\partial D_r^j }{\partial V_i^j}<0\quad\forall r\neq i.
    \end{equation}
    We then choose a vector $\mathbf{x}=[x_1,\cdots,x_n]^\top\in\mathbb{R}^n_{++}$ with the components set to the reciprocal of the total occupied time: ${x}_r=(w_r^j+T_r^j)^{-1}$ $\forall r$. The $i$-th component of the resulting vector $\mathbf{x}^\top\mathbf{J}^j_{\tilde{\mathbf{N}}}(\tilde{\mathbf{V}}^j)$ is given by:
\begin{equation}
     \begin{aligned}
    & \;x_i\cdot\frac{\partial N_i^j}{\partial V_i^j}+\sum_{r\neq i} x_r\cdot\frac{\partial N_r^j}{\partial V_i^j}\\
         = & \;x_i\left(1+ \frac{\partial O_i^j}{\partial V_i^j}\right)+\sum_{r\neq i} x_r\cdot\frac{\partial O_r^j}{\partial V_i^j}\\
         = & \; x_i\left(1+ D_i^j\frac{\partial w_i^j}{\partial V_i^j}\right)+ \sum_{r} x_r(w_r^j+T_r^j)
         \frac{\partial D_r^j}{\partial V_i^j}\\
         = & \;\frac{1}{w_i^j+T_i^j}\left(1+ D_i^j\frac{\partial w_i^j}{\partial V_i^j}\right)+\sum_{r}\frac{\partial D_r^j}{\partial V_i^j}.
     \end{aligned}
\end{equation}
Since the constraint $\mathbf{N} \in \mathcal{X}$ enforces ${V_i^j>V_{i,\mathrm{inv}}^j}$, we have
\begin{equation*}
    1+ D_i^j\frac{\partial w_i^j}{\partial V_i^j}>0.
\end{equation*}
Meanwhile, the Logit demand model implies
\begin{equation}
\begin{aligned}
    & \frac{\partial D_\mathrm{max}^j}{\partial V_i^j}=0=\frac{\partial D_\mathrm{PT}^j}{\partial V_i^j} +\frac{\partial\sum_r D_r^j}{\partial V_i^j}\\ \Rightarrow\quad
    & \sum_{r}\frac{\partial D_r^j}{\partial V_i^j}=-\frac{\partial D_\mathrm{PT}^j}{\partial V_i^j} > 0.
\end{aligned}
\end{equation}
Therefore, $\left[\mathbf{x}^\top\mathbf{J}^j_{\tilde{\mathbf{N}}}(\tilde{\mathbf{V}}^j)\right]_i>0$ $\forall i$, indicating $\mathbf{J}^j_{\tilde{\mathbf{N}}}(\tilde{\mathbf{V}}^j)^\top\mathbf{x}>\mathbf{0}$ and the invertibility of $\mathbf{J}_{\mathbf{N}}(\mathbf{V})$. Consequently, by the Inverse Function Theorem, $\mathbf{J}_\mathbf{V}(\mathbf{N})=\mathbf{J}_{\mathbf{N}}(\mathbf{V})^{-1}$ is well-defined and guaranteed to exist. Additionally, since $\mathbf{J}_{\mathbf{N}}(\mathbf{V})\in\mathbf{Z}$, its inverse $\mathbf{J}_\mathbf{V}(\mathbf{N})$ is a positive matrix.

\bibliographystyle{IEEEtran}
\bibliography{mybibfile}

\end{document}